\documentclass{mfat}
\pagespan{117}{136}

\usepackage{mmap}
\usepackage[utf8]{inputenc}

 \newtheorem{thm}{Theorem}[section]
 \newtheorem{cor}[thm]{Corollary}
 \newtheorem{lem}[thm]{Lemma}
 \newtheorem{prop}[thm]{Proposition}
 \theoremstyle{definition}
 
 \theoremstyle{remark}
 \newtheorem{rem}[thm]{Remark}
 \newtheorem{example}[thm]{Example}
 \numberwithin{equation}{section}

 \newcommand{\be}{\begin{equation}}
\newcommand{\ee}{\end{equation}}
\newcommand{\ba}{\begin{eqnarray}}
\newcommand{\ea}{\end{eqnarray}}
\newcommand{\baa}{\begin{eqnarray*}}
\newcommand{\eaa}{\end{eqnarray*}}
\newcommand{\bb}{\begin{thebibliography}}
\newcommand{\eb}{\end{thebibliography}}
\newcommand{\ci}[1]{\cite{#1}}
\newcommand{\bi}[1]{\bibitem{#1}}
\newcommand{\lab}[1]{\label{eq: #1}}
\newcommand{\re}[1]{(\ref{eq: #1})}
\newcommand{\rep}[1]{{\bf \ref{#1}}}
\newcommand{\rel}[1]{{\bf \ref{#1}}}

\newcommand\al{\alpha}
\newcommand\bt{\beta}
\newcommand\ga{{\bf {\bf \gamma}}}
\newcommand\om{\omega}
\renewcommand\t{\tilde}
\newcommand\dl{\delta}
\newcommand\te{\otimes}
\newcommand\ep{\varepsilon}
\newcommand\la{\lambda}
\newcommand\st{\stackrel{\textstyle{\te}}{,}}

\def\one{\mathbf 1}

\newtheorem{pr}[thm]{Proposition}

\theoremstyle{remark}

\theoremstyle{definition}

\theoremstyle{definition}

\newenvironment{proo}[1][Proof]{\begin{trivlist}
\item[\hskip \labelsep {\bfseries #1}]}{\qed \end{trivlist}}

\def\A{{\mathcal A}}
\def\AA{{\mathfrak A}}
\def\B{{\mathcal B}}
\def\C{{\mathbb C}}
\def\CC{{\mathcal C}}
\def\D{{\mathcal D}}
\DeclareMathSymbol{\DDelta}{\mathalpha}{letters}{"01}
\def\GR{{\mathbb A}}
\def\EE{{\mathcal E}}
\def\F{{\mathbb F}}
\def\FF{{\mathcal F}}
\def\G{{\mathcal G}}
\def\GG{{\mathfrak G}}
\def\H{\mathcal H}
\def\K{{\mathbb K}}
\def\L{{\mathcal L}}
\def\N{\mathcal N}
\def\O{{\mathcal O}}
\def\P{{\mathcal {\mathbf P}}}
\def\Q{{\mathbb Q}}
\def\R{{\mathbb R}}
\def\PP{{\mathfrak P}}
\def\RR{{{\mathcal R}}}
\def\RS{{(\R\setminus 0)}}
\def\T{{\mathbb T}}
\def\TE{{\mathcal T}}
\def\pTE{\widetilde{\mathcal T}}
\def\U{{\mathcal U}}
\def\W{{\mathcal W}}
\def\X{{\mathcal X}}
\def\Y{{\mathcal Y}}
\def\Z{{\mathbb Z}}

\def\b{\mathfrak b}
\def\bb{{\bf b}}
\def\cc{{\bf c}}
\def\d{\mathbf d}
\def\e{{\bf e}}
\def\el{\text{\eurs l}}
\def\er{\text{\eurs r}}
\def\g{\mathfrak g}
\def\gg{\mathfrak g}
\def\gr{\mathcal E}
\def\gs{\ge}
\def\h{\mathfrak h}
\def\ii{{\mathbf i}}
\def\jj{\mathbf j}
\def\k{\mathfrak k}
\def\kk{\mathbf k}
\def\ls{\le}
\def\m{{\bf m}}
\def\n{\mathfrak n}
\def\one{\mathbf 1}
\def\p{{\mathbf p}}
\def\q{\mathbf q}
\def\r{{\bf r}}
\def\w{{\bf w}}
\def\wB{{\widetilde{B}}}
\def\wM{{\widetilde{M}}}
\def\wN{{\widetilde{N}}}
\def\wR{{\widetilde{R}}}
\def\wX{{\bar X}}
\def\hM{{\widehat{M}}}
\def\hN{{\widehat{N}}}
\def\hR{{\widehat{R}}}
\def\hP{{\widehat{P}}}
\def\wx{{\widetilde{\bf x}}}
\def\wz{{\widetilde{\bf z}}}
\def\x{{\bf x}}
\def\y{{\bf y}}
\def\z{{\bf z}}
\def\wh{\widehat}

\def\Ad{\operatorname{Ad}}
\def\Edge{\operatorname{Edge}}
\def\End{\operatorname{End}}
\def\Fl{\operatorname{Fl}}
\def\Gamm{\mathfrak P}
\def\GGG{S_k(n)}
\def\GGR{S^{\R}_k(n)}
\def\Hom{\operatorname{Hom}}
\def\Id{{\operatorname {Id}}}
\def\Jac{{\operatorname {Jac}}}
\def\Lie{{\operatorname {Lie}\;}}
\def\Mat{\operatorname{Mat}}
\def\Net{\operatorname{Net}}
\def\Norm{\operatorname{Norm}}
\def\Poi{{\{\cdot,\cdot\}}}
\def\Proj{\operatorname{Proj}}
\def\Rat{\operatorname{Rat}}
\def\Res{\operatorname{Res}}
\def\Span{{\operatorname{span}}}
\def\Spec{\operatorname{Spec}}
\def\Tr{\operatorname{tr}}
\def\Trace{\operatorname{tr}}
\def\Vert{\operatorname{Vert}}

\def\ad{\operatorname{ad}}
\def\cf{{\operatorname{cf}}}
\def\codim{\operatorname{codim}}
\def\corank{\operatorname{corank}}
\def\deg{{\operatorname{deg}}}
\def\diag{\operatorname{diag}}
\def\dim{\operatorname{dim}}
\def\grad{\mbox{grad}}
\def\id{\operatorname{id}}
\def\ind{\operatorname{ind}}
\def\pr{{\operatorname{pr}}}
\def\rank{\operatorname{rank}}
\def\s{\operatorname{sign}}
\def\sp{{\operatorname{sp}}}
\def\sgn{{\operatorname{sgn}}}
\def\sign{\mbox{sign}}
\def\tr{\operatorname{tr}}
\def\val{\operatorname{val}}

\def\cC{{\A(\wB)}}
\def\:{{:\ }}

\begin{document}

\title[Inverse moment problem for non-Abelian Coxeter double Bruhat cells]
      {Inverse moment problem for non-Abelian Coxeter double Bruhat cells}

\author{Michael Gekhtman}
\address{Department of Mathematics, University of Notre Dame, Notre Dame,
IN 46556, USA}
\email{mgekhtma@nd.edu}

\date{04/05/2015; \ \  Revised 09/02/2016}
\dedicatory{Dedicated to my teacher Yuri Makarovich Berezanskii on his 90th birthday}

\begin{abstract}
We solve the inverse problem for non-Abelian Coxeter double Bruhat cells in terms of the matrix Weyl functions. This result can be
used  to establish complete integrability of the non-Abelian version of nonlinear Coxeter-Toda  lattices in $GL_n$.
\end{abstract}

\maketitle

\section{Introduction}

A fruitful interaction between the operator theory, inverse spectral problems in particular, and the theory of completely integrable systems is, by now, well documented.
Beyond just  linearizing Hamiltonian equations of interest in mathematical physics, this interaction led to greater insight into geometric properties of underlying
objects as well as revealed deep connections with representation theory and algebraic
combinatorics.

In one of the first and most famous instances of an interplay between the spectral theory and integrability questions, Moser \cite{moser} used a map from finite Jacobi matrices to the space
rational functions of fixed degree to linearize the celebrated Toda lattice in the finite non-periodic
case. This map associates with a Jacobi matrix a certain
matrix element of its resolvent, called the Weyl function. On the other hand, the Atiyah-Hitchin Poisson structure \cite{AH} on rational functions initially discovered in the theory of magnetic monopoles, provides a convenient description for the
(linear) Hamiltonian structure of the Toda lattice.

In \cite{FG1}--\cite{FG3}, it was shown that the Atiyah-Hitchin structure belongs to a family of compatible
Poisson structures that can be used to establish a multi-Hamiltonian nature of the entire class
of "Toda-like" integrable lattices. In the context of the linear Poisson structure, these lattices are associated with minimal irreducible co-adjoint orbits of the Borel subgroup in $gl_n$, while, from
the point of view of the quadratic Poisson structure, they are naturally associated with  certain
class of double Bruhat cells in $GL_n$ and belong to the family of so-called Coxeter-Toda lattices.
The latter perspective recently led to establishing of a cluster algebra structure in the space of rational
functions \cite{GSV}. Along with Poisson brackets from \cite{FG2}, the key ingredient of this construction
was a solution of the inverse problem, that allows to restore the Lax operator of a Coxeter-Toda lattice
from its Weyl function in terms of a certain collection of Hankel determinant built from coefficients of the Laurent expansion of the Weyl function. These determinantal formulae generalize the classical ones
in the theory of orthogonal polynomials on the real line and on the unit circle.

In this paper,  we present an overview of a non-Abelian version of some of the results of \cite{FG1}--\cite{FG3} and \cite{GSV}.
Although we will concentrate on finite non-Abelian lattices, it should be pointed out that infinite non-Abelian lattices of Toda type have also attracted a lot of interest of the years.
The have been studied in a variety of contexts and via a variety of approaches, using inverse spectral problems in the semi-infinite case \cite{ BGS}, \cite{BerMokh}, inverse scattering in the double-infinite case \cite{BMRL} and methods of algebraic geometry in the periodic case \cite{K}.

In the earlier paper \cite{GeKo}, we introduce a matrix-valued version of Coxeter-Toda lattices on certain classes of block Hessenberg matrices. These nonlinear lattices generalize both
the nonlinear lattices in \cite{FG1} and the finite non-periodic non-Abelian Toda lattice. We established that
a matrix analogue of the Weyl function provides a convenient tool for a study of these non-Abelian Coxeter-Toda lattices.
In the case of the non-Abelian Toda lattice, this point of view was advocated in \cite{G}. The lattices of
\cite{GeKo} "live" on noncommutative analogues of {\em elementary Toda orbits} -- minimal irreducible co-adjoint orbits of the Borel subgroup in $GL_n$.
In contrast, here we will be more concerned with a noncommutative version of Coxeter double Bruhat cells, whose scalar counterparts are minimal irreducible Poisson submanifolds
of $GL_n$ equipped with the standard Poisson-Lie structure.

In section 3, we define non-Abelian Coxeter double Bruhat cells, introduce some related combinatorial objects and describe how the elements of non-Abelian Coxeter double Bruhat cells can
be parametrized using factorization into elementary factors or, alternatively, using planar directed weighted networks with noncommutative weights.
In section~4, the main section of the paper, we presents a solution of the inverse moment problem for non-Abelian Coxeter double Bruhat cells. Here the key role
is played by the matrix Weyl function. The main theorem, Theorem \ref{invthm}, extends both the results in the commutative case \cite{GSV} and the partial
results obtained in \cite{GeKo}. The factorization parameters are restored as noncommutative monomial expressions in term of Schur complements (quasideterminants)
associated with a family of block Hankel matrices built from the coefficients of the Laurent expansion of the matrix Weyl functions. These quasideterminants replace
ratios of Hankel determinants needed to express the solution of the inverse problem in the commutative case.

In section 5, we show how this inverse problem combined with the Poisson structure on matrix-valued rational functions introduced earlier in \cite{G} lead
to a completely integrable system on every non-Abelian Coxeter double Bruhat cell. We call this system a non-Abelian Coxeter-Toda lattice. The obtained family
of integrable lattices incorporates as particular cases all the lattices from \cite{FG1, GSV, GeKo}.


\section{Preliminaries}

We start by introducing notations and terms to be used throughout the paper.
In what follows we will be dealing with {\em block vectors and block matrices} whose entries
are $m\times m$ matrices.  For an $n_1\times n_2$ block matrix $A=(a_{ij})$, the notation $A^T$ will be reserved for its $n_2\times n_1$ {\em block transpose} :
$A^T=(a_{ji})$.

Denote by $\one_r$ the $r\times r$ identity matrix. Sometimes, when the dimension of the identity matrix is clear from context, we will drop the subscript and use $\one$ instead.

Define  elementary block vectors $e_j= (\delta_{ij} \one_m)_{i=1}^n ,(j=1,\dots,n)$ and elementary block matrices ${e}_{ij}\otimes \one_m=(\delta_{i\al}\delta_{j\bt}\one_m)_{\al,\bt=1}^{n}$.

If $P(\lambda) = \sum_\alpha \lambda^{\alpha} \pi_\alpha$ is a Laurent polynomial with $m\times m$ matrix coefficients, $X$ is an $n\times n$ block matrix and $f$ is a block column vector, we denote by $ P(X) f$ the expression $ \sum_\alpha X^{\alpha}  f \pi_\alpha$ and by
$ f^T P(X) $ the expression $ \sum_\alpha \pi_\alpha f^T X^{\alpha}$.

In what follows, when we deal with an inverse of a block matrix $A=(a_{ij})$, the notation
$A_{ij}^{-1}$ is used for an $(i,j)$-block of $A^{-1}$, while $a_{ij}^{-1}$ or $(A_{ij})^{-1}$ will denote the inverse
of the $(i,j)$-block of $A$.

Recall that if $A=(a_{ij})_{i,j=1}^{2}$ is a $2\times 2$ block matrix (not necessarily with square blocks)
and if a block $A_{ij}$ is square, then its \emph{Schur complement} is defined as
$$
A_{3-i,3-j}^{\square}=A_{3-i,3-j}-A_{3-i,i}(A_{ij})^{-1}A_{j,3-j}\ .
$$
Below, we will use the following well-known
\begin{lem}
\label{Ginv}
Let  $G=\left[ \begin{array}{cc}
A&B \\
C&D
\end{array}\right]$ be an invertible block matrix, whose block $A$ (resp. $B, C, D$) is square
and has an invertible Schur complement. Then $G^{-1}$ is given by the formula
\be
G^{-1}=\left[ \begin{array}{cc}
A^{-1}+A^{-1}B(D^{\square})^{-1}CA^{-1} & -A^{-1}B(D^{\square})^{-1} \\
-(D^{\square})^{-1}CA^{-1} & (D^{\square})^{-1}
\end{array}\right] \ ,\lab{inv1}
\ee
resp.
\be
G^{-1}=\left[ \begin{array}{cc}
 -C^{-1}D(B^{\square})^{-1}& C^{-1}+C^{-1}D(B^{\square})^{-1}AC^{-1} \\
(B^{\square})^{-1} &- (B^{\square})^{-1}AC^{-1}
\end{array}\right] \ ,\lab{inv2}
\ee
\be
G^{-1}=\left[ \begin{array}{cc}
 (A^{\square})^{-1}& -(A^{\square})^{-1}BD^{-1} \\
-D^{-1}C(A^{\square})^{-1}&D^{-1}+D^{-1}C(A^{\square})^{-1}BD^{-1}
\end{array}\right] \ ,\lab{inv3}
\ee
\be
G^{-1}=\left[ \begin{array}{cc}
 -(C^{\square})^{-1}DB^{-1}&  (C^{\square})^{-1}\\
B^{-1}+B^{-1}A(C^{\square})^{-1}DB^{-1}&-B^{-1}A(C^{\square})^{-1}
\end{array}\right] \ .\lab{inv4}
\ee
\end{lem}

\begin{rem} \label{zeroschur}
It is easy to see that if the second row of a block matrix $G=\left[ \begin{array}{cc}
A&B \\
C&D
\end{array}\right]$ with an invertible square $A$ is a left multiple of the first row, then
the Schur complement of $A$ in $G$ is zero.
\end{rem}

For $r\in \mathbb{N}$, denote by $[r]$ the set $\{1,\ldots,r\}$. Given an $n_1\times n_2$ block matrix $A=(a_{ij})$ with $m\times m$ blocks and index sets $I\subset [n_1], J\subset [n_2]$ we denote by $A_I^J$ its block submatrix formed by block rows and columns indexed by $i$ and $J$ resp. For $i\in [n_1], j\in [n_2]$, we denote by $\hat{i} ,\hat{j}$ their complements in $[n_1], [n_2]$.

Following \cite{GGRW}, we denote by $a_{ij}^\square$ the Schur complement of $A_{\hat{i}}^{\hat{j}}$ in $A$
(called a {\em quasideterminant} in terminology of \cite{EGR, GGRW}) :
\be
A_{ij}^{\square}=a_{ij}-A_{i}^{\hat{j}}(A_{\hat{i}}^{\hat{j}})^{-1}A_{\hat{i}}^{j}\ .
\label{Schur}
\ee

\section{Non-Abelian Coxeter double Bruhat cells}
\label{sec:cdbc}

\subsection{}
In this section, we describe combinatorial notions and parameterizations associated with  Coxeter double Bruhat cells adapted to the non-Abelian situation.
The discussion here follows that in sect. 3 of \cite{GSV}. Most of the auxiliary combinatorial statements from that paper can be used without any modifications.

Recall \cite{FZ_JAMS} that a double Bruhat cell $G^{u,v}$ in $GL_n$ associated with a pair of elements $u,v$ of the permutation group
$S_n$ is defined as an intersection
\be
\label{doubleBruhat}
G^{u,v} = \left (\B^+  u \B^+\right ) \cap \left (\B^-  v \B^-\right )\ ,
\ee
where $\B^\pm$ denote subgroups of upper and lower triangular invertible matrices in $GL_n$ and where $u, v$ are identified with the corresponding permutation matrices
$\bar u, \bar v$. Double Bruhat cells in simple and reductive  Lie groups where comprehensively studied
in \cite{FZ_JAMS} in connection with the notion of {\em total positivity}. They also served as a chief motivation for defining the notion of {\em cluster algebras}, as well as an important example of a class of algebraic varieties supporting a cluster algebra structure \cite{BFZ, GSV-book}.

A non-Abelian version of double Bruhat cells was studied in \cite{BR}. For our purposes, they can be defined by \eqref{doubleBruhat} in which $\B^\pm$ now
denote groups of invertible $n\times n$ upper and lower block triangular matrices with $m\times m$ blocks and $u, v$ are now  identified with block  permutation matrices
$\bar u\otimes \one_m, \bar v \otimes \one_m$.

Factorization of generic elements of $G^{u,v}$ into a product of elementary factors plays an important role in the study of double Bruhat cells in both commutative and noncommutative contexts.
For an $m\times m$ matrix $a$ and $i\in [n]$, define {\em elementary block-matrices} $E^+_{i}(a)$, $E^-_{i}(a)$ by
\begin{equation}
E^+_{i}(a) = \one_n\otimes\one_m + e_{i,i+1}\otimes a, \quad  E^-_{i}(a) = \one_n\otimes\one_m + e_{i+1,i}\otimes a\ .
\label{elemA}
\end{equation}
In other words, $E^+_{i}(A)$ (resp. $E^-_{i}(A)$) is a block bidiagonal matrix with $\one_m$ on the diagonal and the only nonzero off-diagonal block $A$ in a position
$(i,i+1)$ (resp. $(i+1,i)$).

As in the scalar case (see, e.g. \cite{Fallat, FZ_JAMS, FZ_MI, GSV}), it will be convenient to represent block matrices that can be realized as products of elementary ones
by planar weighted directed diagrams. In our case, all the weights will be invertible $m\times m$ matrices assigned to edges of the network. Any network $\N$ in question can be drawn in a rectangle, with $n$ sources located on the left side
and $n$ sinks on the right side. Both sources and sinks are labeled
$1$  to $n$  going from the bottom to the top. All internal vertices are trivalent and are either colored white if it has exactly  one incoming edge or black if there are exactly two incoming edges.
There are three kinds of edges: horizontal, directed left-to-right and two kinds of inclined edges, directed southwest or northwest. For a directed path $\mathbf P$ joining $i$th source with the
$j$th sink we define its weight $w({\mathbf P})$ as a left-to-right  ordered product of egde weights in $w({\mathbf P})$. A block matrix $A=A(\N)=(a_{ij})_{i,j=1}^n$ associated with $\N$ is defined
by
\be
a_{ij}= \sum_{{\mathbf P} : i \rightsquigarrow j} w({\mathbf P})\ .
\label{boundary}
\ee
For example, a
$n\times n$ block diagonal  matrix $\diag(d_1,\ldots,d_n)$ and elementary
block bidia\-gonal matrices $E^-_i(l)$ and  $E^+_j(u) $ correspond to planar networks shown in Figure~\ref{fig:factor} a), b) and c), respectively; all weights not shown explicitly are equal to~1.


\begin{figure}[ht]
\begin{center}
\includegraphics[height=3.0cm]{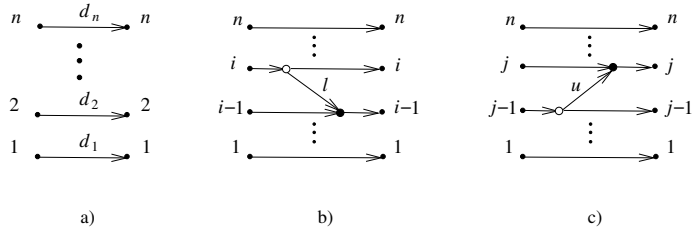}
\caption{Elementary networks}
\label{fig:factor}
\end{center}
\end{figure}


Two networks of the kind described above can be concatenated by gluing the sinks of the former to the sources of the latter.
If $A_1$, $A_2$ are matrices associated with the two networks, then it is clear
that the matrix associated with their concatenation
is $A_1 A_2$.

\subsection{}
Recall that {\em a Coxeter element} of $S_n$ is any element of length $n-1$ or, in other words, a
of all $n-1$ distinct elementary transpositions $s_i\  (i=1,\ldots, n-1)$ taken in an arbitrary order.
We are only interested in non-Abelian double Bruhat cells associated with a pair of Coxeter elements $u,v$,
{\em Coxeter double Bruhat cells} for short.

Denote $s_{[p,q]}=s_ps_{p+1}\ldots s_{q-1}$
for $1\le p< q \le n$ and recall that every Coxeter element $v\in S_n$ can be written in the form
\begin{equation}
\label{factoru}
v=
s_{[i_{k-1} , i_k]}\cdots s_{[i_{1} , i_2]}s_{[1 , i_1]}
\end{equation}
for some subset $I=\{1=i_0 < i_1 < \cdots < i_k=n\}\subseteq[1,n]$. Besides,
define $L=\{1=l_0 < l_1 < \cdots < l_{n-k}=n\}$ by $\{ l_1 < \cdots < l_{n-k-1}\} = [1,n] \setminus I$.

\begin{lem}\label{u_inv}
Let $v$ be given by~{\rm \eqref{factoru}}, then
$$
v^{-1} =
s_{[l_{n-k-1} , l_{n-k}]}\cdots s_{[l_{1} , l_2]}s_{[1 , l_1]}\ .
$$
\end{lem}

Let $(u,v)$ be a pair of Coxeter elements and
\begin{equation}\label{IL}
 \begin{aligned}
 I^+= & \, \{1=i^+_0 < i^+_1 < \cdots < i^+_{k^+}=n\}, \\
 I^-= & \, \{1=i^-_0 < i^-_1 < \cdots < i^-_{k^-}=n\},\\
 L^+= & \, \{1=l^+_0<  l^+_1 < \cdots < l^+_{n-k^+-1}<  l^+_{n-k^+}= n\},\\
 L^-= & \, \{1=l^-_0<  l^-_1 < \cdots < l^-_{n-k^--1}<  l^-_{n-k^-}= n\}
 \end{aligned}
 \end{equation}
  be subsets of $[1,n]$ that correspond to $v$ and $u^{-1}$ in the way just described.

Certain additional combinatorial data that was utilized in the commutative case \cite{GSV} can also be employed in a non-Abelian situation.
Namely, given a pair $(u,v)$ of Coxeter elements or, equivalently,  the sets $I^\pm$ given by~\eqref{IL},
we define, for any $i\in [1,n]$  integers $\varepsilon^\pm_i$ and $\zeta^\pm_i$:
\begin{equation}
\varepsilon^\pm_i=\left \{ \begin{array}{ll} 0, & \mbox{if}\  i=i^\pm_j\
\mbox{for some}\  0 < j \leq k_\pm \, ,
\\ 1, & \mbox{otherwise}\end{array}
\right.
\label{eps}
\end{equation}
and
\begin{equation}
\zeta^\pm_i =i (1-\varepsilon^\pm_i) -\sum_{\beta=1}^{i -1} \varepsilon^\pm_\beta;
\label{nunu}
\end{equation}
note that by definition, $\varepsilon^\pm_1=1$, $\zeta^\pm_1=0$.
Further, put
\begin{equation}
M^\pm_i=\{\zeta^\pm_\alpha\ : \ \alpha=1,\ldots,i \}
\label{Mi}
\end{equation}
and
\begin{equation}
k^\pm_i=\max \{ j: i^\pm_j \leq i\}.                               
\label{k+/-}
\end{equation}
Finally, define
\begin{equation}
\varepsilon_i= \varepsilon_{i}^+ + \varepsilon_{i}^-               
\label{epsum}
\end{equation}
and
\begin{equation}
 \varkappa_i= i+1 - \sum_{\beta=1}^i \varepsilon_\beta.            
 \label{kappa}
 \end{equation}


\begin{lem}
\label{allcomb}
{\rm (i)} The $n$-tuples $\varepsilon^\pm=(\varepsilon^\pm_i)$ and $\zeta^\pm=(\zeta^\pm_i)$ uniquely
determine each other.

{\rm (ii)} For any $i\in [1,n]$,
 \begin{equation*}
\zeta^\pm_i =
\left \{ \begin{array}{ll} j, & \mbox{if}\  i=i_j^\pm\ \mbox{for
some}\  0< j \leq k_\pm\, ,
\\ -\sum_{\beta=1}^{i -1} \varepsilon^\pm_\beta, &
\mbox{otherwise} .\end{array} \right .
\end{equation*}

{\rm (iii)} For any $i\in [1,n]$,
 $$
 k^\pm_i=i-\sum_{\beta=1}^i \varepsilon^\pm_\beta,\quad \varkappa_i= k_i^+ + k_i^- - i +1\ .
$$

{\rm (iv)} For any $i\in [1,n]$,
\begin{equation*}
M^\pm_i= [ k^\pm_i -i+1, k^\pm_i ]= \Big[1 - \sum_{\beta=1}^i
\varepsilon^\pm_\beta, i-\sum_{\beta=1}^i \varepsilon^\pm_\beta \Big].
\end{equation*}
\end{lem}

A set of $m\times m$ complex matrices $c_1^-,\ldots, c_{n-1}^-; c_1^+,\ldots, c_{n-1}^+; d_1, \ldots, d_n$ will play a role
of noncommutative parameters for generic elements in $G^{u,v}$. We will call $c_i^-$ {\em lower}, $c_i^+$ {\em upper},
and $d_i$ {\em diagonal} factorization parameters.

Define matrices $D=\diag (d_1,\ldots, d_n)$,
\begin{equation}\label{Cj}
C^+_j\!=\! \sum_{\alpha=i^+_{j-1}}^{i^+_j-1}  {e}_{\alpha,\alpha+1}\otimes c^+_\alpha,\quad j\in [1, k^+],\quad
C^-_j\!=\! \sum_{\alpha=i^-_{j-1}}^{i^-_j-1}  {e}_{\alpha+1,\alpha}\otimes c^-_\alpha,\quad j\in [1,  k^-],
\end{equation}
and
$$
\bar C^+_j \!=\! \sum_{\alpha=l^+_{j-1}}^{l^+_j-1}  {e}_{\alpha,\alpha+1}\otimes c^+_\alpha,\quad j\in [1, n-k^+],\quad
\bar C^-_j\!=\! \sum_{\alpha=l^-_{j-1}}^{l^-_j-1} {e}_{\alpha+1,\alpha}\otimes c^+_\alpha ,\quad j\in [1, n-k^-].
$$

\begin{lem}
\label{X-Xinv}
A generic  element $X\in G_m^{u,v}$ can be written as
\begin{equation}
X=  (\one -C^-_1)^{-1} \cdots  (\one - C^-_{k^-})^{-1} D (\one - C^+_{k^+})^{-1}\cdots (\one -
C^+_1)^{-1},
\label{factorI}
\end{equation}
and its inverse can be factored as
\begin{equation}
X^{-1}=
(\one + \bar C^+_{n-k^+})^{-1}\cdots (\one +
\bar C^+_1)^{-1} D^{-1} (\one +
\bar C^-_1)^{-1} \cdots  (\one + \bar C^-_{k^-})^{-1}.
\label{factorinv}
\end{equation}
\end{lem}

The network $N_{u,v}$ that corresponds to factorization (\ref{factorI}) is obtained by the concatenation (left to right)
of $2 n -1$ building blocks (as depicted in Fig.~\ref{fig:factor}) that correspond to elementary matrices
\begin{align*}
&E^-_{i^-_{2}-1}(c^-_{i^-_{2}-1}), \ldots,  E^-_{1}(c^-_{1}), E^-_{i^-_{3}-1}(c^-_{i^-_{3}-1}),\ldots,
E^-_{i^-_{2}}(c^-_{i^-_{2}}), \ldots, \\
& E^-_{n-1}(c^-_{n-1}), \ldots,   E^-_{i^-_{k^--1}}(c^-_{i^-_{k^--1}}), D,
E^+_{i^+_{k^+-1}}(c^+_{i^+_{k^+-1}}), \ldots,  E^+_{n-1}(c^+_{n-1}),\\
&\ldots, E^+_{i^+_{2}}(c^+_{i^+_{2}}),\ldots  E^+_{i^+_{3}-1}(c^+_{i^+_{3}-1}),
E^+_{1}(c^+_{1})\cdots  E^+_{i^+_{2}-1}(c^+_{i^+_{2}-1})\, .
\end{align*}
This network has $4(n-1)$ internal vertices and $5n-4$ horizontal edges.

Similarly, the network $\bar N_{u,v}$ that corresponds to factorization (\ref{factorinv})
is obtained by the concatenation (left to right)
of building blocks  that correspond to elementary matrices
\begin{align*}
&E^+_{i^+_{2}-1}(-c^+_{i^+_{2}-1}),\ldots, E^+_{1}(-c^+_{1}),E^+_{i^+_{3}-1}(-c^+_{i^+_{3}-1}),
\ldots, E^+_{i^+_{2}}(-c^+_{i^+_{2}}),\ldots ,\\
& E^+_{n-1}(-c^+_{n-1}),\ldots, E^+_{i^+_{k^+-1}}(-c^+_{i^+_{k^+-1}}),
D^{-1}, E^-_{i^-_{k^--1}}(-c^-_{i^-_{k^--1}}), \ldots, E^-_{n-1}(-c^-_{n-1}),  \\
&\ldots, E^-_{i^-_{2}}(-c^-_{i^-_{2}}), \ldots, E^-_{i^-_{3}-1}(-c^-_{i^-_{3}-1}),
E^-_{1}(-c^-_{1}),\ldots, E^-_{i^-_{2}-1}(-c^-_{i^-_{2}-1})\, .
\end{align*}

\begin{rem}\label{cmvao}
(i) If $v=s_{n-1} \cdots s_1$, then $X$ is a block lower Hessenberg matrix, and if
$u=s_{1} \cdots s_{n-1}$, then $X$ is a block upper Hessenberg matrix.

(ii) If
$v=s_{n-1} \cdots s_1$ {\em and} $u=s_{1} \cdots s_{n-1}$, then $G^{u,v}$ consists of
block tri-diagonal matrices with non-zero off-diagonal entries ({\em block Jacobi matrices\/}).
In this case $I^+=I^-=[1,n]$, $\varepsilon_1^\pm=1$ and $\varepsilon_i^\pm=0$ for $i=2, \ldots, n$.

(iii) If
$u=v=s_{n-1} \cdots s_1$ (which leads to $I^+=[1,n], I^-=\{1,n\}$), then, in the scalar case,
elements of $G^{u,v}$
have a structure of recursion operators arising in the theory of
orthogonal polynomials on the unit circle.

(iv) The choice $u=v=( s_{1} s_3 \cdots ) (s_2 s_4 \cdots )$ (the so-called {\em
bipartite Coxeter element\/})
gives rise to a special kind of pentadiagonal block matrices $X$. In the scalar case, they are called {\em CMV
matrices\/}) and serve
as an alternative version of recursion operators for orthogonal polynomials
on the unit circle \cite{cmv} and in the complex plane  \cite{BerIvMokh}.
\end{rem}

\subsection{Example}
\label{runex}

 Let $n=5$, $v =s_4 s_3 s_1 s_2$ and $u =s_3 s_2 s_1 s_4$. The network $N_{u,v}$ that corresponds
to factorization (\ref{factorI}) is shown in Figure~\ref{factorex}.


\begin{figure}[ht]
\begin{center}
\includegraphics[height=4.0cm]{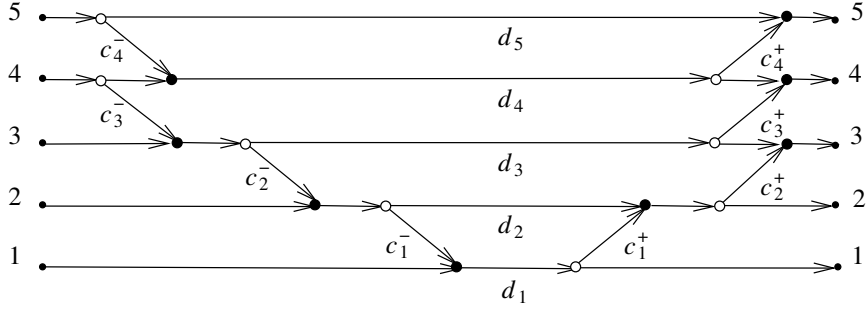}
\caption{Network representation for elements in $G^{s_3 s_2 s_1 s_4,s_4 s_3 s_1 s_2}$}
\label{factorex}
\end{center}
\end{figure}


A generic element $X\in G^{u,v}$ has a form
$$
X= (x_{ij})_{i,j=1}^5=\left (
\begin{array}{ccccc}
d_1 &  x_{11}c_1^+ & x_{12}c_2^+ & 0 & 0\\
c_1^-x_{11} & d_2+c_1^-x_{12} & x_{22}c_2^+ & 0 & 0\\
c_2^-x_{21} & c_2^-x_{22} & d_3+c_2^-x_{23} & d_3c_3^+ & 0\\
c_3^-x_{31} & c_3^-x_{32} & c_3^-x_{33} & d_4+c_3^-x_{34} & d_4c_4^+\\
0 & 0 & 0 & c_4^-d_4 & d_5+c_4^-x_{45}
\end{array}
\right ).
$$

One finds by a direct observation that $k^+=3$ and $I^+=\{i_0^+,i_1^+,i_2^+,i_3^+\}=\{1,3,4,5\}$, and hence
$L^+=\{l_0^+,l_1^+,l_2^+\}=\{1,2,5\}$. Next, $u^{-1}=s_4s_1s_2s_3$, therefore, $k^-=2$ and
$I^-=\{i_0^-,i_1^-,i_2^-\}=\{1,4,5\}$, and hence $L^-=\{l_0^-,l_1^-,l_2^-,l_3^-\}=\{1,2,3,5\}$.
Further,
$$
\varepsilon^+=(1,1,0,0,0), \quad\varepsilon^-=(1,1,1,0,0),
$$
 and hence
 $$
 \zeta^+=(0,-1,1,2,3),\quad \zeta^-=(0,-1,-2,1,2).
 $$
Therefore,
$$
(k^+_i)_{i=1}^5=(0,0,1,2,3), \quad (k^-_i)_{i=1}^5=(0,0,0,1,2),
$$
and hence
\begin{align*}
(M^+_i)_{i=1}^5&=([0,0], [-1, 0], [-1,1], [-1,2], [-1,3]),\\
(M^-_i)_{i=1}^5&=([0,0], [-1, 0], [-2,0], [-2,1], [-2,2]).
\end{align*}
Finally, $\varepsilon=(2,2,1,0,0)$ and
$\varkappa=(0,-1,-1,0,1)$.

The network $\bar N_{u^{-1},v^{-1}}$ that corresponds
to factorization (\ref{factorinv}) is shown in Figure~\ref{factorinvex}.


\begin{figure}[ht]
\begin{center}
\includegraphics[height=4.0cm]{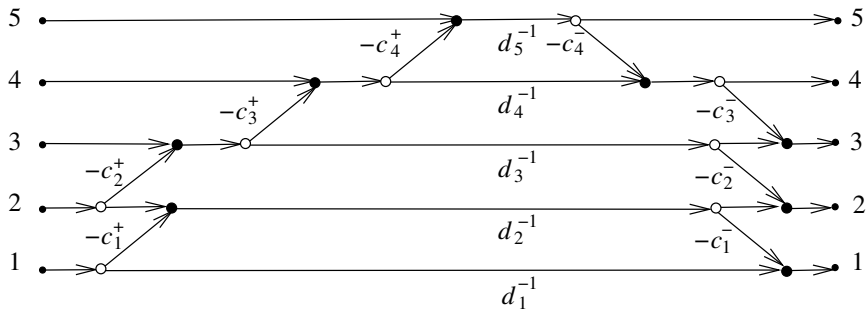}
\caption{Network $\bar N_{u^{-1},v^{-1}}$ for the double Bruhat cell  $G^{s_3 s_2 s_1 s_4,s_4 s_3 s_1 s_2}$}
\label{factorinvex}
\end{center}
\end{figure}


\section{Inverse problem}\label{sec:invprob}

\subsection{} With each $X\in G^{u,v}$ one associates a \emph{matrix Weyl function}
\be M(\la)=M(\la,X)=e_1^T(\la - X)^{-1} e_1=\sum_{k=0}^{\infty}\frac{1}{\la^{k+1}}h_k, \lab{Weyl} \ee
where 
\be\label{moments}
h_k=e_1^T X^ke_1
\ee
are the \emph{moments} of $X$.

Our goal is to show how a generic element $X$ of a non-Abelian Coxeter double Bruhat cell $G^{u,v}$ that admits factorization
(\ref{factorI}) can be restored from its Weyl function \re{Weyl} up to a block-diagonal  conjugation preserving the Weyl function.
We denote by $G^{u,v}/{\mathbf T}$ the space of orbits of this action on $G^{u,v}$. Here ${\mathbf T}$ denotes the group of invertible block diagonal matrices of the form $T=\diag \left (\one_m, T_1, T_{n-1}) \right )$.

If $X\in G^{u,v}$ has factorization parameters $c_i^\pm, d_i$ and $A$ is a block diagonal matrix $T=\diag \left (\one_m, c^-_1,\ldots, (c^-_{n-1}\cdots c_{1}^-) \right )$, then
$X':= T^{-1} X T$ has all lower parameters equal to $\one_m$ and its upper and diagonal parameters are given by
\begin{equation}
{\bf c}_i = (c^-_{i-1}\cdots c_{1}^-)^{-1} c_i^+ c_i^- (c^-_{i-1}\cdots c_{1}^-),\quad   {\bf d}_i = (c^-_{i-1}\cdots c_{1}^-)^{-1} d_i (c^-_{i-1}\cdots c_{1}^-)\ .
\label{paramC}
\end{equation}

Clearly,  the Weyl function of $X'$ coincides with that of $X$ and we can view \eqref{paramC} as parameterizing a generic element of $G^{u,v}/{\mathbf T}$. In other words, the inverse
problem we are interested in solving can be restated as follows: given an element in $X\in G^{u,v}$ with all lower parameters equal to $\one$, restore the remaining factorization parameters
${\bf c}_i, {\bf d}_i$ from the Weyl function $M(\lambda, X)$.

To solve the inverse problem, we combine the approach employed in the commutative situation \cite{FG3, GSV} with the one
used in a non-Abelian setting in the  block Jacobi case \cite{krein, Berezan} (see also \cite{G0, shmisha}) and block Hessenberg ($v= s_{n-1}\cdots s_1$) case \cite{GeKo}. The main idea stems from the classical moments problem
\cite{akh}:  in the commutative case,
one considers the space
$\mathbb{C}[\lambda, \lambda^{-1}] /\det (\lambda - X) $ equipped with
the so-called {\em moment functional} - a bi-linear functional $\langle\ , \ \rangle$ on Laurent polynomials in one variable, uniquely defined by the property
 \begin{equation}
 \label{momfun}
 \langle \lambda^i, \lambda^j \rangle = h_{i+j}.
 \end{equation}
$X$ is then realized as a matrix of the  operator of multiplication by $\lambda$ relative to  appropriately
selected  bases $ \{p_i^+(\lambda)\}_{i=0}^{n-1}$, $\{p_i^-(\lambda)\}_{i=0}^{n-1}$ bi-orthogonal with respect to
the moment functional:
\begin{equation}
\langle p_i^-(\lambda), p_j^+(\lambda) \rangle = \delta_{ij}.
\label{matmom}
\end{equation}
For example, the classical tridiagonal case corresponds to the orthogonalization of the sequence $1, \lambda, \ldots, \lambda^{n-1}$.  Elements of $G^{s_{n-1}\cdots s_1, s_{n-1}\cdots s_1}$ (cf. Remark~\ref{cmvao}(iii)) result from the bi-orthogonalization
of sequences $1, \lambda, \ldots, \lambda^{n-1}$ and $\lambda^{-1}, \ldots, \lambda^{1-n}$, while
CMV matrices (Remark~\ref{cmvao}(iv)) correspond to the bi-orthogonalization
of sequences $1, \lambda, \lambda^{-1}, \lambda^{2}, \ldots $ and   $1, \lambda^{-1}, \lambda, \lambda^{-2},
\ldots$
The non-Abelian case was first treated in pioneering works by
M.~G. Krein \cite{krein} and Yu.~M. Berezanskii \cite[chapter VII.2]{Berezan} on Jacobi matrices with matrix (operator) valued coefficients. In this case, the moment functional defined by \eqref{momfun} becomes a matrix-valued functional that acts on pairs of matrix-valued Laurent polynomials $a(\lambda)= \sum_{i}\lambda^i a_i, b(\lambda)= \sum_{j}^N \lambda^i b_j$ by
$$
\langle a(\lambda), b(\lambda \rangle = \sum_{i,j} a_i h_{i+j} b_j\
$$
and, in the tri-diagonal case, coefficients of $X$ are obtained via so-called {\em pseudo-orthogo\-nalization}
\cite{Berezan} applied to the sequence $\one, \lambda \one, \ldots
$
We will generalize this strategy to the case of arbitrary Coxeter $u,v$.

For any $l\in \mathbb Z$, $i\in \mathbb N$ define
block Hankel matrices
\begin{equation}
\H^{(l)}_i=(h_{\alpha +\beta + l - i - 1})_{\alpha,\beta=1}^i\, .
\label{hank}
\end{equation}
Matrices $\H^{(l)}_i$ play the key role in the solution of the inverse problem. Before describing its solution, let us recall the situation in the scalar case which can be summarized
in the following theorem quoted from \cite{GSV}.

\begin{thm}
 \label{invthmscalar} Let $\Delta_i^{(l)}=\det \H^{(l)}_i$.  If $X\in G^{u,v}$ admits factorization~{\rm\eqref{factorI}}, then
\begin{equation}\label{cd}
\begin{aligned}
{\bf d}_i  =d_i &=\frac{\Delta_i^{(\varkappa_i + 1)}\Delta_{i-1}^{(\varkappa_{i-1})}}{\Delta_i^{(\varkappa_i)}\Delta_{i-1}^{(\varkappa_{i-1} + 1)}},\\
{\bf c}_i  =c_i^+c_i^-&= \frac{\Delta_{i-1}^{(\varkappa_{i-1})}\Delta_{i+1}^{(\varkappa_{i+1})}}{\left (\Delta_i^{(\varkappa_{i}+1)}\right )^2}
\left ( \frac{\Delta_{i+1}^{(\varkappa_{i+1} + 1)}}{\Delta_{i+1}^{(\varkappa_{i+1})}}  \right )^{\varepsilon_{i+1}}
\left ( \frac{\Delta_{i-1}^{(\varkappa_{i-1} + 1)}}{\Delta_{i-1}^{(\varkappa_{i-1})}}  \right )^{2-\varepsilon_{i}}
\end{aligned}
\end{equation}
for any $i\in [1,n]$.
\end{thm}

The rest of this section is devoted to the proof of the noncommutative analogue of Theorem \ref{invthmscalar} that can be formulated as

\begin{thm}
 \label{invthm}
 If $X$ is  a generic element of $G^{u,v}$ admitting factorization~{\rm\eqref{factorI}} then corresponding noncommutative factorization parameters
${\bf c}_i, {\bf d}_i$ can be restored as noncommutative monomial expressions in terms of quasideterminants associated with corner block entries of matrices
$\H^{(l)}_i$.
\end{thm}

\subsection{} The following short-hand notations will be convenient for us below: for any integers 
$r < s$ introduce block column vectors
$h^{[r,s]}=\mbox{col}[h_{r},h_{r+1}, \ldots, h_{s}]$
and a block row vectors $h_{[r,s]}=[h_{r},h_{r+1}, \ldots, h_{s}]$. 
For example, we can partition ${\H}_{i+1}^{(l)}$ as
\begin{equation}\label{parts}
\begin{aligned}
{\H}_{i+1}^{(l)} &= \left [\begin{array} {cc} {\H}_{i}^{(l-1)} & h^{[l,l+i-1]}\\ h_{[l,l+i-1]}  & h_{l+i} \end{array}\right ] =
\left [\begin{array} {cc} h_{l-i} & h_{[l-i+1,l]}\\ h^{[l-i+1,l]}  & {\H}_{i}^{(l+1)} \end{array}\right ] \\
&=
 \left [\begin{array} {cc}  h^{[l-i,l-1]} & {\H}_{i}^{(l)} \\ h_{l} & h_{[l+1,l+i]}   \end{array}\right ] =
\left [\begin{array} {cc} h_{[l-i,l-1]} & h_{l} \\ {\H}_{i}^{(l)}  & h^{[l+1,l+i]}  \end{array}\right ] \ .
\end{aligned}
\end{equation}
Using \eqref{Schur}, \eqref{parts} and  Lemma \ref{Ginv}, we can express "corner" block entries of the inverse of ${\H}_{i+1}^{(l)}$ using quasideterminants:
\begin{equation}\label{quasi_inv}
\begin{aligned}
({\H}_{i+1}^{(l)})^{-1}_{11}& =  \left (({\H}_{i+1}^{(l)})_{11}^\square\right)^{-1}\!=\! \left (h_{l-i} -
 h_{[l-i+1,l]} ({\H}_{i}^{(l+1)})^{-1} h^{[l-i+1,l]}   \right
 )^{-1} ,
 \\
({\H}_{i+1}^{(l)})^{-1}_{1,i+1}& =  \left (({\H}_{i+1}^{(l)})_{1,i+1}^\square\right)^{-1}\!=\! \left (h_{l} -
 h_{[l-i,l-1] }({\H}_{i}^{(l)})^{-1} h^{[l+1,l+i]}   \right
 )^{-1} ,
 \\
({\H}_{i+1}^{(l)})^{-1}_{i+1,1}& =  \left (({\H}_{i+1}^{(l)})_{i+1,1}^\square\right)^{-1}\!=\! \left (h_{l} -
 h_{[l+1,l+i]} ({\H}_{i}^{(l)})^{-1} h^{[l-i,l-1]}    \right
 )^{-1} ,
 \\
({\H}_{i+1}^{(l)})^{-1}_{i+1,i+1}& =  \left (({\H}_{i+1}^{(l)})_{i+1,i+1}^\square\right)^{-1}\!=\! \left (h_{l+i} -
 h_{[l, l+i-1]} ({\H}_{i}^{(l-1)})^{-1} h^{[l,l+i-1]}   \right )^{-1} .
\end{aligned}
\end{equation}

Note that in the scalar case all expressions above are ratios of Hankel determinants. This is precisely the reason why Theorem \ref{invthm} serves as an noncommutative
analogue of  Theorem \ref{invthmscalar}. Although we are not going to use them below, we should mention the following identities that  were proved in \cite{GeKo}.

\begin{prop} \label{lemT}
For $k\geq 1$

\begin{align}
& (\H_{k+1}^{(l)})_{k+1,k+1}^{-1}  =-(\H_{k+1}^{(l)})_{k+1,1}^{-1}\left ((\H_{k}^{(l-1)} )_{k,1}^{-1}\right )^{-1}
(\H_{k}^{(l)})_{k,k}^{-1}\, , \lab{Hi1}
\\
& (\H_{k+1}^{(l)})_{11}^{-1}  =-(\H_{k}^{(l)})_{11}^{-1}\left ((\H_{k}^{(l+1)})_{k1}^{-1}\right )^{-1}
(\H_{k+1}^{(l)})_{k+1,1}^{-1}\, , \lab{H1i}
\\
& (\H_{k+1}^{(l)})_{k+1,k+1}^{-1} \big ( (\H_{k+1}^{(l)})_{1,k+1}^{-1} \big )^{-1}  =
- (\H_{k}^{(l)})_{k k}^{-1}\big ( (\H_{k}^{(l-1)})_{1 k}^{-1} \big )^{-1} \, . \lab{H1more}
\end{align}

\end{prop}

Next, we will use relations between moments \eqref{moments} to define a matrix  analogue of the characteristic polynomial for $X$.
\begin{lem}
For a generic $X\in \mathcal{H}$, the Weyl function can be factored as
\be  M(\la)=Q(\la)P^{-1}(\la)
\ , \lab{MPQ}\ee
where $P(\la), Q(\la)$
are monic matrix polynomials with $n\times n$ matrix coefficients of degrees $n, n-1$ resp.  In particular,
$$P(\la)=\la^{n}\one-\la^{n-1} F_{n-1}- \cdots - F_0\ .$$
\label{factorWeyl}
\end{lem}
\begin{proof}
Equation  \re{MPQ} is equivalent to relations
\be
\lab{momdep}
h_{k+n} = \sum_{j=0}^{n-1} h_{k+j} F_j\quad (k=0,1,\ldots)\, ,
\ee
which, in turn, can be re-written as
\be
\label{CharPolyCoeffs}
\H^{(l+1)}_n = \H^{(l)}_n
\left[ \begin{array}{cccc}
0 & & \cdots &F_0\\
\one & \ddots & & \vdots\\
&\ddots& & \\
&&\one&F_{n-1}
\end{array}\right]
\ee
for all  $l \in \mathbb{Z}$. The latter follows from \eqref{hank} and the fact that, for a generic $X$, block column $X^n e_1$ is a linear combination over $gl_m$ of
block columns $e_1, X e_1,  \ldots, X^{n-1} e_1$: $X^n e_1= e_1 F_0 + \cdots + X^{n-1} e_1 F_{n-1}$.
\end{proof}

\begin{rem}
\label{companion} It is a corollary of the proof above that $X$ is similar to the block companion matrix that appears in the right hand side of \eqref{CharPolyCoeffs}.
This means, in particular that the characteristic polynomial of $X$ coincides with $\det P(\lambda)$.

\end{rem}

If $u, v, \ldots $ are block row or column vectors,  we will denote by $\Span_{gl_m}\{u,v,\ldots \}$
the space of all combinations of $A_u u + A_v v + \cdots $ with $m\times m$ matrix coefficients
$A_u, A_v, \ldots $
 For any $i\in [n]$ define subspaces
 $$
 \L^+_i=\Span_{gl_m} \{e^T_1,\dots,e^T_i\},\quad \L^-_i=\Span_{gl_m}  \{e_1,\ldots, e_i\}.
 $$
Let
 \begin{equation}
{\displaystyle
{\bf {\bf \gamma}}^+_{i} = \begin{cases} \vec{\prod}_{\beta=0}^{j-1}
\left ( {\mathbf d}_{i^+_\beta} {\mathbf c}_{i_\beta^+}
\cdots {\mathbf c}_{i^+_{\beta+1}-1}\right )\, , &\mbox{if }
i=i_j^+\, ,
\\
(-1)^{l^+_\alpha -1} \vec{\prod}_{\beta=0}^{\alpha-1} \left ( {\mathbf c}_{l^+_\beta}\cdots {\mathbf c}_{l^+_{\beta+1}-1}
{\mathbf d}^{-1}_{l^+_\beta+1} \right )\, , & \mbox{if } i=l_\alpha^+ \, ,\end{cases}
}
\label{gamma_+}
\end{equation}
 \begin{equation}
{\displaystyle
{\bf {\bf \gamma}}^-_{i} = \begin{cases}  {\mathbf d}_{i^-_{j-1}} \cdots {\mathbf
d}_{i^-_0}\, ,
&\mbox{if } i=i_j^- \, , \\
(-1)^{l^-_\alpha -1} {\mathbf d}^{-1}_{l^-_\alpha} \cdots  {\mathbf d}^{-1}_{l^-_1}\, ,  & \mbox{if }
i=l_\alpha^- ,\ i < n\, , \end{cases}
}
\label{gamma_-}
\end{equation}
where $i^\pm_j, l^\pm_\alpha$ are defined in \eqref{IL} and ${\bf {\bf
\gamma}}^\pm_1=1$\, .

\begin{rem}
\label{monom_gamma}
Note an expression for ${\bf {\bf \gamma}}_i^+$ has a form ${\bf {\bf \gamma}}_i^+= \mu\ {\mathbf c}_{i-1}$ if $\varepsilon_i^+=0$ or ${\bf {\bf \gamma}}_i^+= \mu\  {\mathbf c}_{i-1} {\mathbf d}_{i}^{-1} $
if $\varepsilon_i^+=1$, where in both cases $\mu$ is a noncommutative monomial in ${\mathbf c}_{\alpha}\ (\alpha < i -1)$ and ${\mathbf d}^{\pm 1}_{\beta}\ (\beta < i )$.
This means, in particular, that all weights ${\mathbf c}_{1}, \ldots, {\mathbf c}_{n-1}$ can be uniquely restored from ${\mathbf d}_{1}, \ldots, {\mathbf d}_{n}$ and
 ${\bf {\bf \gamma}}_1^+, \ldots, {\bf \gamma}_n^+$ as noncommutative monomial expressions.
\end{rem}

 \begin{lem}
\label{flaglemma} For any $i\in [1, n]$ one has
\begin{equation}
 e^T_1 X^{\zeta^+_i}  = {\bf \gamma}^+_{i} e^T_{i}  \mod \L^+_{i-1} \label{flag1}
\end{equation}
and
\begin{equation}
  X^{\zeta^-_i} e_1 =e_i {\bf \gamma}^-_{i}   \mod \L^-_{i-1} .    \label{flag2}
\end{equation}
In particular,
$$
\L^+_i=\Span_{gl_m} \{e^T_1 X^{\zeta^+_1},\dots,e^T_1 X^{\zeta^+_i}\}, \quad
\L^-_i=\Span_{gl_m} \{ X^{\zeta^-_1}e_1,\dots, X^{\zeta^-_i}e_1\}.
$$
\label{lemmaflag}
In addition,
\begin{equation}
X^{\zeta^-_n-n} e_1 = (-1)^{n -1} e_{n} \left ({\mathbf d}^{-1}_{n} {\mathbf d}^{-1}_{l^-_{n-k^--1}}\cdots  {\mathbf d}^{-1}_{l^-_1} \right ) :=   (-1)^{n -1} e_{n} \tilde{\mathbf \gamma}_n   \mod \L^-_{n-1}.
\label{gamma_n_alt}
\end{equation}
\end{lem}
 \begin{proof} The proof given in \cite{GSV} relies only on the combinatorics of networks associated with $X, X^{-1}$ and thus can be applied
 to the noncommutative case as well. The idea is that to analyze $X^k e_1$ (resp. $e_1^T X^k$) for some $k\in \mathbb{Z}$ , one can consider a network obtained by concatenation of $|k|$ copies of the network that corresponds to $X^{\sign k}$ and find what is the highest horizontal level that can be reached
by a path starting at the lowest level on the right (resp. on the left). The upshot is that if $k=\zeta_i^-$ (resp. $k=\zeta_i^+$) then the level is given by $i$ and the corresponding path is unique. The weight of this path is equal to ${\bf \gamma}_i^-$ (resp. ${\bf \gamma}_i^+$).
  \end{proof}

\begin{example}{\rm  We illustrate (\ref{flag1}) using Example~\ref{runex} and Fig.~\ref{factorex}. If $j>0$ then to find $i$
such that $e_1^T X^j = {\bf \gamma}^+_i e^T_i \mod \L^+_{i-1}$ it is enough to find the highest sink
that can be reached by a path starting from the source 1 in the network obtained by concatenation of
$j$ copies of $N_{u,v}$. Thus, we conclude from Fig.~\ref{factorex}, that
\begin{align*}
e_1^T X &= d_1 c_1^+ c_2^+ e^T_3 \mod \L^+_{2},\quad e_1^T X^2 = d_1 c_1^+ c_2^+ d_3 c_3^+ e^T_4 \mod \L^+_{3},\\
e_1^T X^3 &= d_1 c_1^+ c_2^+ d_3 c_3^+ d_4 c_4^+ e^T_5 \mod \L^+_{4}.
\end{align*}
Similarly, using the network  $\bar N_{u^{-1},v^{-1}}$ shown in Fig.~\ref{factorinvex}, one observes
that $e_1^T X^{-1} = -c_1^+ d_2^{-1}  e^T_2 \mod \L^+_{1}$. These relations are in agreement
with (\ref{flag1}).
}
\end{example}

\begin{lem}
The matrix $F_0$ in \re{momdep} has an expression
$$
F_0=(-1)^{n-1}  {\mathbf d}_{l^-_1} \cdots  {\mathbf d}_{l^-_{n-k^--1}} {\mathbf d}_{n} {\mathbf d}_{i^-_{k^--1}} \cdots  {\mathbf d}_{i^-_{1}} {\mathbf d}_{1}\ .
$$
\label{block-det}
\end{lem}

\begin{proof} Multiplying the equation $X^n e_1= e_1 F_0 + \cdots + X^{n-1} e_1 F_{n-1}$ on the left by $X^{\zeta^-_n-n}$ we obtain
$$X^{\zeta^-_n} e_1= X^{\zeta^-_n-n}e_1 F_0 + \cdots + X^{\zeta^-_n-1} e_1 F_{n-1}\ .$$
From definitions \eqref{eps}--\eqref{kappa} and Lemma \ref{allcomb} we conclude
that $M_{n-1}^- = \{\zeta^-_\alpha\ : \ \alpha=1,\ldots,n-1 \} = \left [ \zeta^-_n -n +1, \zeta^-_n-1\right ]$. By Lemma \ref{flaglemma}, this means that in the equality above
only the left hand side and the first term in the right hand side have non-zero last block entries, equal to $ {\bf \gamma}^-_n$ and $(-1)^{n -1} \tilde{\bf \gamma}^-_n F_0$, resp.
Therefore, $F_0= (-1)^{n -1} \left(\tilde{\bf \gamma}^-_n\right )^{-1} {\bf \gamma}^-_n$, which proves the claim.
\end{proof}

\begin{lem} Let $k\in [1,n-1]$ and $X_i$ be the $k\times k$ block submatrix of $X\in G^{u,v}$ obtained by deleting $n-k$ last block rows and columns.
Then
\begin{equation}\label{submomrange}
h_\alpha (X_k)= h_\alpha(X)
 \end{equation}
for $\alpha\in [\varkappa_k  - k+1, \varkappa_k  + k]$.
\label{submoment}
\end{lem}

\begin{proof} The proof in the scalar case was given in \cite{GSV}. It depends only on combinatorics of networks $N_{u,v}$ and $\bar N_{u^{-1}, v^{-1}}$ and
thus translates to the non-Abelian case without any changes.
\end{proof}





\subsection{}
The solution of the inverse problem relies
on properties of matrix polynomials of the form
\ba
\label{polyhank}
{\P}^{(l)}_i(\lambda)&=\left [\begin{array}{cccc}
h_{l-i+1} &h_{l-i+2} &\cdots & h_{l+1}
\\ \cdots &\cdots &\cdots &\cdots
\\
h_l & h_{l+1} & \cdots &h_{l+i}
\\ \one & \lambda\one
&\cdots & \lambda^{i}\one
\end{array}\right ]_{i+1,i+1}^\square \\
\nonumber
&= \lambda^i -
 \left [ \one \ \lambda\one  \ \cdots \lambda^{i-1}\one \right ] \left ( \H^{(l)}_i\right )^{-1}
h^{[l+1,l+i]}.
\ea

\begin{lem}
Expand $\P^{(l)}_i(\lambda)$ as $\P^{(l)}_i(\lambda)=\sum_{\alpha=0}^{i} \lambda^{i} \pi_\alpha$.
Then
\begin{equation}\label{rowrepeat}
\begin{aligned}
 e_1^T X^\beta \P^{(l)}_i(X) e_1=&\sum_{\alpha=0}^{i} h_{\alpha+\beta}\pi_{\alpha}=0\ \quad \mbox{for}
 \quad \beta \in [l-i, l-1],
 \\
e_1^T X^{l+1} \P^{(l)}_i(X) e_1=&
\sum_{\alpha=0}^{i} h_{\alpha+l+1}\pi_{\alpha}
=\left ( \H^{(l+1)}_{i+1}\right )^\square_{i+1,i+1}\ ,
\\
e_1^T X^{l-i}  \P^{(l)}_i(X) e_1=&
\sum_{\alpha=0}^{i} h_{\alpha+l-i}\pi_{\alpha}=
\left ( \H^{(l)}_{i+1}\right )^\square_{1,i+1}\ .
\end{aligned}
\end{equation}
\end{lem}

\begin{proof} The first equality in all three lines follows from \eqref{polyhank} and \eqref{moments}. It is easy to see that, for $\beta \in [l-i, l-1]$,
$
 \sum_{\alpha=0}^{i} h_{\alpha+\beta}\pi_{\alpha}
 $ is equal to the $(i+1,i+1)$-quasideterminant of the block matrix obtained from $\H^{(l+1)}_{i+1}$
 by replacing the last block row with the block row number $\beta + i -l$. By Remark \ref{zeroschur}, this quasideterminant is equal to zero.
 This prove the first equality in this Lemma. The other two follow from \eqref{polyhank} combined with \eqref{quasi_inv}.
 \end{proof}

 \begin{cor}
 \label{CharPoly}
 Let $\P(\lambda)$ be the matrix polynomial defined in Lemma \ref{factorWeyl}.
 Then, for any $l\in \mathbb{Z}$,
 $$
 \P(\lambda)= \P^{(l)}_n(\lambda)\ .
 $$
 \end{cor}

 \begin{proof} The claim follows from comparing \re{momdep} with the first equation in Lemma \ref{rowrepeat}.
 \end{proof}

 \begin{lem}\label{atzero}
$$
 \P^{(l)}_i(0) = - ({\H}_{i}^{(l)})^{-1}_{11} \left ( ({\H}_{i}^{(l+1)})^{-1}_{1i} \right )^{-1}=- \left ( ({\H}_{i}^{(l)})^{\square}_{11}\right )^{-1}   ({\H}_{i}^{(l+1)})^{\square}_{1i}\ .
 $$
 \end{lem}

\begin{proof} By \eqref{polyhank}, $ \P^{(l)}_i(0) = - \left [ \one \ 0  \ \cdots \ 0 \right ] \left ( \H^{(l)}_i\right )^{-1} h^{[l+1,l+i]}$.
Using \re{inv3}, the latter expression can be rewritten as
$
- ({\H}_{i}^{(l)})^{-1}_{11} \left ( h_{l+1} -   h_{[l-i+1,l]} ({\H}_{i-1}^{(l+1)})^{-1} h^{[l+2,l+i]}   \right )
$, which proves the claim when one takes into consideration the second equality in \eqref{quasi_inv}.
 \end{proof}

 \begin{lem}\label{recover_d} For $k \in [1,n]$,
$$
 \P^{(\varkappa_k)}_k(0) = (-1)^{k-1}  \left ( \overrightarrow{\prod}_{1<l^-_\alpha<k}{\mathbf d}_{l^-_\alpha} \right )
 {\mathbf d}_{k} \left (\overleftarrow{\prod}_{1\leq i^-_\beta<k}{\mathbf d}_{i^-_{\beta}} \right )\, .
 $$
 \end{lem}

\begin{proof} Consider the matrix $X_k$ as defined in  Lemma \ref{submomrange}. Since
$h_\alpha (X_k)= h_\alpha(X)$
for $\alpha\in [\varkappa_k  - k+1, \varkappa_k  + k]$, $k\times k$ block Hankel matrices
${\H}_{k}^{(\varkappa_k)}$ and  ${\H}_{k}^{(\varkappa_k +1)}$ can be chosen to play the same role for $X_k$ that
$n\times n$ block Hankel matrices
${\H}_{n}^{(l)}$ and  ${\H}_{n}^{(l +1)}$ play in the equation \eqref{CharPolyCoeffs} for $X$. By Corollay \ref{CharPoly} this implies, in turn,
that matrix polynomial $\P^{(\varkappa_k)}_k(\lambda)$ plays the same role for $X_k$ that  $\P^{(l)}_n(\lambda)$ does for $X$. In particular, we can apply
Lemma \ref{block-det} to express $\P^{(\varkappa_k)}_k(0)$ using diagonal weights of $X_k$ which coincide with the first $k$ diagonal weights of $X$.
The claim then follows from Lemma \ref{block-det}.
 \end{proof}

\begin{prop}
 \label{polycorr}
 Define matrix Laurent polynomials
 $$
 {\bf p}_{i}(\lambda)
 = (-1)^{(i-1)\varepsilon_i^-}\lambda^{k^-_i - i + 1} \P^{(\varkappa_{i-1}-\varepsilon_i^-)}_{i-1}(\lambda)
\left ( \P^{(\varkappa_{i-1}-\varepsilon_i^-)}_{i-1}(0)\right )^{-\varepsilon_i^-}({\bf \gamma}_i^-)^{-1}
 ,\quad i\in [1,n].
 $$
 Then
 $$
 {\bf p}_{i}(X)e_1= e_{i},\quad i\in [1,n].
 $$
 \end{prop}

\begin{proof} Expand $\P^{(\varkappa_{i-1}-\varepsilon_i^-)}_{i-1}(\lambda)$ as
$\P^{(\varkappa_{i-1}-\varepsilon_i^-)}_{i-1}(\lambda) = \sum_{\alpha=0}^{i-1} \lambda^{\alpha} \pi_{\alpha}$. Then Lemma \ref{rowrepeat}
implies that
\be\label{polycorr1}
 \sum_{\alpha=0}^{i-1} h_{\alpha + \beta}\pi_{\alpha}=0
\ee
 for $\beta \in [\varkappa_{i-1} - \varepsilon_{i}^- - i +1, \varkappa_{i-1} - \varepsilon_{i}^- -1]$.

It follows from Lemma \ref{flaglemma} that
 $$
  e_i {\bf \gamma}_i^- = X^{\zeta_i^-} e_1 +  \sum_{\alpha=1}^{i-1} X^{\zeta^-_\alpha} e_1 \zeta_{\alpha}
 $$
 for some $m\times m$ coefficients $\zeta_\alpha$. By Lemma~\ref{allcomb}(iv), this can be re-written as
  \be\label{polycorr2}
 e_i {\bf \gamma}_i^- =  X^{k^-_i - i + 1} \sum_{\alpha=0}^{i-1} X^{\alpha-1} e_1 \tilde\pi_{\alpha},
 \ee
 where either $ \tilde\pi_{i-1}=1$ (if $\varepsilon_i^-=0$), or
 $ \tilde\pi_0=1$ (if $\varepsilon_i^-=1$). Define a matrix polynomial ${\bf p}(\lambda)=\sum_{\alpha=0}^{i-1} \lambda^{\alpha} \tilde\pi_{\alpha}$.
 By Lemma \ref{flaglemma}, block vectors $e_1^T X^\beta$, $\beta\in M^+_{i-1}$, span  $\L^+_{i-1}$ over $gl_m$. Therefore, by Lemma~\ref{allcomb}(iv),
 \be\label{polycorr3}
 0=e_1^T X^\beta  X^{k^-_i - i + 1}  {\bf p}(X) e_1= \sum_{\alpha=0}^{i-1} h_{\alpha+\beta +k^-_i - i+1}\tilde\pi_{\alpha}
 \ee
 for $\beta \in [k_{i-1}^+ - i +2, k_{i-1}^+]$. Using  Lemma \ref{allcomb} (iii), we conclude that $k^-_i= k^-_{i-1} + 1 -\varepsilon^-_{i}$ and that
 in the equation \eqref{polycorr2} $(\beta +k^-_i - i+1)$ ranges through the interval $\beta \in [\varkappa_{i-1} - \varepsilon_{i}^- - i +1, \varkappa_{i-1} - \varepsilon_{i}^- -1]$.

 Comparing with \eqref{polycorr3}, we see that coefficients of polynomials ${\bf p}(\lambda)$ and
 $\P^{(\varkappa_{i-1}-\varepsilon_i^-)}_{i-1}(\lambda)$ satisfy the same system of linear equations. The genericity assumption guarantees that this system has a unique solution if one requires that either the leading or degree zero coefficient of the resulting polynomial is equal to $\one$. Indeed, the unique solvability of the system in this case relies on invertibility of the block Hankel matrices $\H^{(l-1)}_{i-1}$, $\H^{(l)}_{i-1}$ resp., where
 $l= \varkappa_{i-1} - \varepsilon_{i}^-$. This means that ${\bf p}(\lambda)=\P^{(\varkappa_{i-1}-\varepsilon_i^-)}_{i-1}(\lambda)$ if $\varepsilon_i^-=0$ and
 ${\bf p}(\lambda)=\P^{(\varkappa_{i-1}-\varepsilon_i^-)}_{i-1}(\lambda)\pi_0^{-1}$ otherwise. Then it drops out from \eqref{polycorr2} that
 $$
 e_i =  X^{k^-_i - i + 1} \P^{(\varkappa_{i-1}-\varepsilon_i^-)}_{i-1}(X) e_1 \pi_0^{-\varepsilon_i^-}({\bf \gamma}_i^-)^{-1}\ ,
 $$
 which proves the claim.
 \end{proof}

 \subsection{}
Finally, we can complete the proof of Theorem \ref{invthm}.
Lemma \ref{atzero} and Lemma \ref{recover_d} imply the equation
$$
 {\mathbf d}_{k} = (-1)^k       \left ( \overrightarrow{\prod}_{1<l^-_\alpha<k}{\mathbf d}_{l^-_\alpha} \right )^{-1}
 \left ( ({\H}_{k}^{(\varkappa_k)})^{\square}_{11}\right )^{-1}   ({\H}_{k}^{(\varkappa_k+1)})^{\square}_{1k}
 \left (\overleftarrow{\prod}_{1\leq i^-_\beta<k}{\mathbf d}_{i^-_{\beta}} \right )^{-1}\ ,
$$
which allows to recursively restore $d_1, \ldots, d_n$ as noncommutative monomials in terms of corner quasideterminants
of block Hankel matrices ${\H}_{k}^{(\varkappa_k)}, {\H}_{k}^{(\varkappa_k+1)}$.

Next, observe that by Lemma \ref{flaglemma} and Corollary \ref{polycorr},
\begin{equation*}
\begin{split}
 {\bf {\bf \gamma}}_i^+ &= e_1^T X^{\zeta_i^+} {\bf p}_{i} (X) e_1 \\
& = (-1)^{(i-1)\varepsilon_i^-}\left ( e_1^T X^{\zeta_i^+ + k^-_i - i + 1} \P^{(\varkappa_{i-1}-\varepsilon_i^-)}_{i-1}(X) e_1\right )
\left ( \P^{(\varkappa_{i-1}-\varepsilon_i^-)}_{i-1}(0)\right )^{-\varepsilon_i^-}({\bf {\bf \gamma}}_i^-)^{-1} .
\end{split}
\end{equation*}
 Since by (\ref{nunu}), \eqref{epsum}, (\ref{kappa}) and Lemma~\ref{allcomb}(iii),
 $$
 \zeta_i^++k_i^--i+1=\varkappa_i-(i-1)\varepsilon_i^+\ =
 \begin{cases}
  \varkappa_{i-1}-\varepsilon_i^- +1, & \mbox{if} \quad \varepsilon_i^+=
  0\, ,
  \\
 \varkappa_{i-1}-\varepsilon_i^- -i +1, & \mbox{if} \quad \varepsilon_i^+=
 1\, ,
 \end{cases}
 $$
 we can apply the second (if $\varepsilon_i^+= 0$) or the third (if $\varepsilon_i^+= 1$) equality in Lemma \ref{rowrepeat}
 with $i$ replaced with $i-1$ and $l= \varkappa_i -1$ to conclude that
 $$
 e_1^T X^{\zeta_i^+ + k^-_i - i + 1} \P^{(\varkappa_{i-1}-\varepsilon_i^-)}_{i-1}(X) e_1 =
 \begin{cases}
\left ( \H^{(\varkappa_i)}_{i}\right )^\square_{ii}, & \mbox{if} \quad \varepsilon_i^+=
0\, ,
\\
\ &\ \\
\left ( \H^{(\varkappa_i)}_{i}\right )^\square_{1i}, & \mbox{if} \quad \varepsilon_i^+=
1\, .
 \end{cases}
 $$
 Thus
 $$
 {\bf {\bf \gamma}}_i^+ {\bf {\bf \gamma}}_i^- = - (-1)^{(i-1)\varepsilon_i^-}
 \left ( \H^{(\varkappa_i)}_{i}\right )^\square_{i - (i-1)\varepsilon_i^+, i}
 \left ( \left ( ({\H}_{i}^{(\varkappa_{i-1}-\varepsilon_i^-)})^{\square}_{11}\right )^{-1}
 ({\H}_{i}^{(\varkappa_{i-1}-\varepsilon_i^-+1)})^{\square}_{1i}\right )^{-\varepsilon_i^-}\, .
 $$
 This relation, together with Remark \ref{monom_gamma}, completes the proof of Theorem \ref{invthm}.

\section{Non-Abelian Coxeter-Toda lattices}

\subsection{}
We start this section by reviewing basic facts about
{\em Toda flows\/} on $GL_n$. These are commuting Hamiltonian
flows generated by conjugation-invariant functions on $GL_n$ with respect to  the standard Pois\-son--Lie structure.
Toda flows (also known as {\em characteristic Hamiltonian systems\/}) are defined for an arbitrary standard
semi-simple Poisson--Lie group, but we will concentrate on the $GL_n$ case,
where as a maximal algebraically independent family of conjugation-invariant functions  one can choose
$F_k : GL_n \ni X \mapsto \frac{1}{k}
\tr X^ k$, $k=1,\ldots, n-1$. The equation of motion generated by $F_k$ has a
{\em Lax form}
\begin{equation}
\label{Lax_intro}
dX/dt = \left [ X,\  - \frac{1}{2} \left ( \pi_+(X^k) - \pi_-(X^k)\right ) \right ],
\end{equation}
where $\pi_+(A)$ and $\pi_-(A)$ denote strictly upper and lower parts of a matrix $A$.

Any double Bruhat cell  $G^{u,v}$, $u,v \in S_n$, is a regular Poisson submanifold  in $GL_n$ invariant under the
right and left multiplication by elements of the maximal torus (the subgroup of diagonal matrices) $\TE \subset GL_n$.
In particular,  $G^{u,v}$ is invariant under the conjugation by elements of $\TE$. The standard Poisson--Lie structure
is also invariant under the conjugation action of $\TE$ on $GL_n$. This means that Toda
flows defined by (\ref{Lax_intro}) induce commuting Hamiltonian flows on $G^{u,v} / \TE $ where  $\TE$ acts on
$G^{u,v}$ by conjugation.
In the case when $v=u^{-1}=(n\ 1\ 2 \ldots  n-1)$,
$G^{u,v}$ consists of tridiagonal matrices with nonzero off-diagonal entries,  $G^{u,v} / \H $  can be conveniently described
as the set $\Jac$ of {\em Jacobi matrices} of the form
\[
L =  \left( \begin{array}{ccccc}
b_{1} & a_{1} & 0 & \cdots & 0 \\
1 &b_{2}& a_{2} &\cdots&0\\
&\ddots&\ddots&\ddots&\\
&&&b_{n-1}& a_{n-1} \\
0&&&1 &b_{n}
\end{array}\right), \quad  a_1 \cdots a_{n-1} \ne 0, \quad \det L \ne 0.
\]
Lax equations \eqref{Lax_intro} then become the equations of the {\em finite nonperiodic Toda
hierarchy\/}
$$
dL/dt=[L,-\pi_+(L^k)],
$$
the first of which, corresponding to $k=1$, is the celebrated {\it Toda lattice}
\begin{eqnarray*}
da_j/dt&=&a_j(b_{j+1}-b_j), \quad j=1,\ldots , n-1,\\
\nonumber
db_j/dt&=&(a_j-a_{j-1}), \quad j=1,\ldots ,n,
\end{eqnarray*}
 with the boundary conditions $a_0=a_n=0$. Recall that $\det L$ is a Casimir function for the
 standard Poisson--Lie bracket.
 The level sets of the function $\det L$  foliate $\Jac$ into $2(n-1)$-dimensional symplectic
 manifolds, and the Toda hierarchy defines a completely integrable system on every symplectic leaf.
Note that although Toda flows on an arbitrary double Bruhat cell $G^{u,v}$ can be exactly solved via the so-called
 {\em factorization method\/}, in most cases the dimension of symplectic leaves in
 $G^{u,v} / \TE$ exceeds $2(n-1)$, which means that conjugation-invariant functions do not
 form a Poisson commuting family rich enough to ensure Liouville complete integrability.

 An important role in the study of Toda flows played by the
Weyl function
\begin{equation}
m(\lambda)=m(\lambda;X)=((\lambda\one-X)^{-1} e_1,e_1)=\frac{q(\lambda)}{p(\lambda)},
\label{weyl}
\end{equation}
in the scalar case is well-known.
Here $p(\lambda)$ is the characteristic polynomial of $x$ and $q(\lambda)$ is the characteristic polynomial
of the $(n-1)\times(n-1)$ submatrix of $x$ formed by deleting the first row and column.
Differential equations that describe the evolution of $m(\lambda;X)$ induced by Toda flows do not
depend on the initial value $x(0)$ and are easy to solve: though nonlinear,  they
are also induced by {\em linear differential equations with constant coefficients\/} on the space
\begin{equation}
\label{Rat}
\left \{ \tilde m(\lambda) =\frac{\tilde q(\lambda)}{p(\lambda)}
\ : \deg p = \deg \tilde q+1 = n,\  \text{$p, \tilde q$ are coprime,  $p(0) \ne 0$}
\right \}
\end{equation}
 by the map $ \tilde m(\lambda) \mapsto m(\lambda) = -\frac{1}{h_0} \tilde m(-\lambda)$, where
 $H_0=\lim_{\lambda\to \infty}\lambda \tilde m(\lambda)\ne 0$.

Since $m(\lambda;X)$ is invariant under the action of $\TE$ on $G^{u,v}$ by conjugation, we have
 a map from  $G^{u,v} / \TE $ into the space
$$
\W_n=\left \{ m(\lambda) =\frac{q(\lambda)}{p(\lambda)}
\ : \deg\ p =\deg\ q+1 = n,\  \text{$p, q$ monic and coprime,  $p(0) \ne 0$}
\right \}.
$$
 In the tridiagonal  case, this map is sometimes called the {\em Moser map}. Its inverse is computed via solving the classical
moment problem. In \cite{GSV}, we have shown that for any Coxeter double Bruhat cell

(i)  the Toda hierarchy defines a completely integrable system on level sets of the determinant in $G^{u,v} / \TE$, and

(ii) the Moser map $m_{u,v}:G^{u,v} / \TE\to\W_n$ defined in the same way as in the tridiagonal case is invertible.

\noindent Integrable equation induced on $G^{u,v} / \TE $ by Toda flows are called {\em Coxeter--Toda lattices}.

Since Coxeter--Toda flows associated with different choices of $(u,v)$  lead to the same evolution of the Weyl
function, and the corresponding Moser maps are invertible, one can construct transformations between  different
$G^{u,v} / \TE$ that preserve the corresponding Coxeter--Toda flows
and thus serve as {\em generalized B\"acklund--Darboux
transformations\/}  between them.
The main goal of \cite{GSV} was to describe these transformations from the cluster algebra point
of view. We plan to pursue this line of inquiry in the noncommutative situation. A study of non-Abelian
Coxeter--Toda lattices serves as the first step in this direction.

\subsection{Non-Abelian Toda lattice}
We now return to the non-Abelian case.
As we mentioned in Remark \ref{cmvao} (ii), if
$v=u^{-1}=s_{n-1} \cdots s_1$, then $G^{u,v}$ consists of
 block Jacobi matrices.
Using conjugation by block diagonal matrices preserving the matrix Weyl function, we can ensure that
elements of $G^{u,v}/{\mathbf T}$ are represented by block Jacobi matrices with all subdiagonal blocks equal to $\one$. These form
a subset we denote by $Jac$ in the set $Hess$ of {\em block upper Hessenberg matrices}
of the form
\ba X=\left[ \begin{array}{cccccc}
b_1 & a_1 & \Lambda^{(0)}_1&\Lambda_1^{(1)}&\cdots& \Lambda_ 1^{(n-3)}\\
\one & b_2 &a_2& \Lambda^{(0)}_2& \ddots&\vdots\\
&\one&b_3&a_3& \ddots&\Lambda_{n-3}^{(1)}\\
&&\one&b_4 &\ddots &\Lambda^{(0)}_{n-2}\\
&&&\ddots&\ddots &a_{n-1}\\
&&&&\one&b_n
\end{array}\right] \, .\lab{Hess} \ea



Let $\mathfrak{b}_{> 0},\mathfrak{b}_{\leq 0}$ be Lie algebras of block-strictly upper triangular and block lower triangular $N\times N$ matrices respectively.    We can represent any $n\times n$ block marix
$A$ as
$$A=A_{\leq 0}+A_{> 0}$$
using a decompositions
$$gl_{m n}=\mathfrak{b}_{\leq 0}+\mathfrak{b}_{>0}\ .$$

Following the Adler-Kostant-Symes construction \cite{Ko,Sy}, one identifies  $\mathfrak{b}_{\leq 0}^*$, the dual of
$\mathfrak{b}_{\leq 0}$,  with $Hess$ via the trace form $\langle X, Y \rangle := \mathrm{Tr}(XY)$ and endow  $Hess$ with a linear Poisson structure obtained as a pull-back of the Lie-Poisson (Kirillov-Kostant) structure on $\mathfrak{b}_{\leq 0}^{*}$. Then a Poisson bracket of two scalar-valued functions $f_1, f_2$ on $Hess$ is
\be
\{f_1,f_2\}(X)=\mathrm{Tr}(X,[(\nabla f_1(X))_{\leq 0}, (\nabla f_2(X))_{\leq 0}]) \lab{Kost}
\ee
where gradients are computed with respect to the trace form.

Symplectic leaves of this bracket  are orbits of the co-adjoint action of the group $B_-$ of block lower triangular invertible matrices:
\baa
\mathcal{O}_{J+X_0}=\{Ad_n^*(J+X_0),n \in B_-\}=\{J+\pi_{b_+}\left (Ad_n(X_0)\right );\ n \in B_-, b_+\cong \mathfrak{g}_+^{\perp}\},
\eaa
where $J$ the $n\times n$ block matrix with $\one$s on the block sub-diagonal and zeros everywhere else.

The hierachy of nonabelian Kostant-Toda flows on $\mathcal{H}$ is generated by the Hamiltonians
$$H_k(X)=\frac{1}{k+1}\mbox{Tr}(X^{k+1}),\quad k=1, \ldots, n.$$
Each flow has a Lax form
\be \dot{X}=[X,(X^k)_{\leq 0}]\, . \label{lax}
\ee
The first Hamiltonian in the family above does not depend on blocks $\Delta_i^{(j)}$ in
$X$
 \be
 \label{firsttoda}
 H_1=\frac{1}{2}\mathrm{Tr}(X^2)=\mathrm{Tr}\Big(\sum_{j=1}^{n-1}a_j+\frac{1}{2}\sum_{j=1}^{n}b_j^2\Big)\ .
 \ee
On the subspace $Jac$ of $Hess$ defined by vanishing of all $\Lambda_i^{(j)}$, it induces the following evolution equations on blocks $a_j, b_j$:
\be
\label{NAToda}
\dot{a}_j=a_j b_{j+1}-b_{j}a_j, \quad
\dot{b}_j=a_{j} - a_{j-1}\quad (j=1,\ldots ,n,\ a_0=a_n=0)\ .\\
\ee
These are the equations of the \emph{non-Abelian  Toda lattice}. This exactly solvable system was
first introduced by A. Polyakov as a discretization of the principal chiral field equation. In the doubly-infinite case
for a suitable class of initial data it was solved via the inverse scattering method in \cite{BMRL}. In \cite{K},
the solution in theta-functions was found for the periodic non-Abelian Toda lattice. Semi-infinite and finite non-periodic non-Abelian Toda equations were integrated in \cite{G0}.
Another approach, based on a theory of quasideterminants, was applied in \cite{EGR} to integrate both the finite non-Abelian Toda lattice
and its two-dimensional generalization.

 \subsection{Non-Abelian Moser map and complete integrability}

As we mentioned above, In the scalar case, and for $X$ tridiagonal, the map from $ X$ to its Weyl was used by Moser \cite{moser}
to linearize the finite non-periodic lattice. In \cite{FG1,FG2,GSV}, the Moser map was utilized to study
the multi-Hamiltonian structure for Coxeter-Toda lattices and to construct a cluster algebra structure
in a space of rational functions of given degree. In the block tridiagonal case, the matrix Weyl function was used in \cite{G0,G}
to linearize the non-Abelian Toda lattice and establish its complete integrability. Many of the results
of \cite{G} remain valid in the situation we are considering here and are reviewed below.

For a generic $X\in G^{u,v}$, recall the factorization \eqref{factorWeyl} of its the Weyl function: $M(\lambda, X)=Q(\la)P^{-1}(\la)$.

Denote by $\mathcal{P}$ the permutation operator in $\mathbf{C}^m\times\mathbf{C}^m$ :  $\mathcal{P}(x\te y)=y\te x.$
We summarize properties of the non-Abelian Moser map in the following
\begin{prop}{\rm(}\cite{G}{\rm)}. \label{PoissM}
\begin{enumerate}
\item[(i)] The Poisson bracket induced by the pushforward of
the Lie-Poisson structure \re{Kost} under the non-abelian Moser map
satisfies
\begin{equation}\lab{MMbrack}
\begin{aligned}
\{M(\la) \st M(\mu)\}= & -\frac{1}{\la-\mu}\big(M(\la)-M(\mu)\big)\te \big(M(\la)-M(\mu)\big)\mathcal{P}
\\
                              & +\big(M(\la) M(\mu)\te M(\la)-M(\mu)\te
                              M(\mu)M(\la)\big)\mathcal{P}.
\end{aligned}
\end{equation}
 \item[(ii)] Polynomial $P(\la)=\la^{n}\one-\la^{n-1} F_{n-1}- \dots - F_0$ is conserved by the flows (\ref{lax}) and the Poisson brackets between the matrix entries of $P(\la)$ are given by
 \be
~\{P(\la) \st P(\mu)\}=\Big[\frac{\textstyle{1}}{\textstyle{\la-\mu}}\mathcal{P},P(\la)\te P(\mu)\Big]  \ .
\lab{PPtbrack}
\ee
\item[(iii)] If $Z$ is a fixed invertible matrix with distinct eigenvalues then coefficients of
the polynomial $\det( Z P(\la) + z \one_m)$ form a maximal family of Poisson commuting functions.
\item[(iv)] Evolution equations induced on moments $h_k$ by the first flow in (\ref{lax}) can be solved explicitly: for $k=0,\ldots, 2n-2$, $h_k=\t{h}_0^{-1}\t{h}_k$, where  $\t{h}_k$ are entries of a block Hankel matrix
$\t{H}=(\t{h}_{i+j})_{i,j=0}^{n-1}$ that solves a linear equation with constant coefficients
$$
\dot{\t{H}}  = \t{H}
\left[ \begin{array}{cccc}
0 & & \cdots &F_0\\
\one & \ddots & & \vdots\\
&\ddots& & \\
&&\one&F_n
\end{array}\right]
$$
and initial conditions $\t{h}_k(0)={h}_k(0)$. For $k \geq 2n-1$, $h_k$ are determined by \re{momdep}.
 \end{enumerate}
\end{prop}
\noindent In \re{MMbrack}, \re{PPtbrack} above we used tensor (\emph{St. Petersburg}) notations for Poisson brackets
of matrix elements of matrix-valued functions
$$\{A \st B\}_{i\,j}^{i'j'}=\{A_{ii'} , B_{jj'}\}$$
(see, e.g. \cite{FT}).

Although Proposition \ref{PoissM} was originally proved in \cite{G}  for non-Abelian Toda lattice, it can be used to define a completely integrable
system on $G^{u,v}/{\mathbf T}$ for an arbitrary pair of Coxeter permutations $u, v$. Indeed, on can endow the space ${\mathcal M}$ of matrix-valued rational functions
admitting factorization \eqref{factorWeyl} with a Poisson structure given by \re{MMbrack}. Part (iii) of Proposition \ref{PoissM} guarantees that coefficients of the bivariate polynomial
$\det( Z P(\la) + z \one_m)$ generate a completely integrable system on ${\mathcal M}$. The inverse problem we solved in section 4 allows us to induce a Poisson structure and completely
integrable Hamiltonian flows on $G^{u,v}/{\mathbf T}$. Note that coefficients of $\det P(\lambda)$ belong to the maximal family of involutive functions we constructed. By Remark \ref{companion}
they generate the algebra of spectral invariants of a generic element $X\in G^{u,v}/{\mathbf T}$. In particular, $H_1=\frac{1}{2}\mathrm{Tr}(X^2) = \frac{1}{2}\mathrm{Tr}(F_{n-1} + F_{n}^2)$
generates a completely integrable nonlinear Hamiltonian flow on $G^{u,v}/{\mathbf T}$ that can be linearized using part (iv) of Proposition \ref{PoissM} and the inverse problem
of section 4. We call this flow a {\em non-Abelian Coxeter-Toda lattice} on $G^{u,v}/{\mathbf T}$.

\subsection{Examples} We conclude with some examples of non-Abelian Coxeter-Toda lattices that correspond to the case when $v=s_{n-1} \cdots s_1$ and thus $X$ is a block upper Hessenberg matrix. These were studied in \cite{GeKo}. In this situation, it is convenient to slightly modify a parametrization of $G^{u,v}/{\mathbf T}$ by requiring that subdiagonal block
entries of $X$ rather than lower weights $c_i^-$ are equal to $\one_m$. This can be achieved via block diagonal conjugation. Indeed, for $v=s_{n-1} \cdots s_1$, all lower $c_i^-$
equal to $\one_m$ and diagonal and upper weights given by ${\bf d}_i$ and ${\bf c}_i$ resp. (cf. \eqref{paramC}), the subdiagonal entries of $X\in G^{u,v}/{\mathbf T}$ can be found
from \eqref{factorI} to be equal to $X_{i+1,i}={\bf d}_i$. The conjugation of $X$ by the block diagonal matrix
$T=\diag \left (\one_m, {\bf d}_1,\ldots, ({\bf d}_{n-1}\cdots {\bf d}_{1}) \right )$
will reduce $X$ to the form we seek. In order not to overload the notations, we retain symbol $X$ for the resulting  matrix and symbols ${\bf d}_i$ and ${\bf c}_i$ for its diagonal and upper weights. Note that Theorem \ref{invthm} still remains valid under this conditions.

Recall that for $v=s_{n-1} \cdots s_1$, $\varepsilon_i^-=0\  (i=2,\ldots n)$. Specializing Lemma \ref{X-Xinv} to this case, one obtains
{\small
\baa X=\left[ \begin{array}{ccccc}
{\bf d}_1 & {\bf d}_1{\bf c}_1 & \ep_2 {\bf d}_1{\bf c}_1{\bf c}_2& \hspace{5mm}\cdots\hspace{5mm}& \ep_2 \ep_3 \cdots \ep_{n-1}{\bf d}_1{\bf c}_1{\bf c}_2\cdots {\bf c}_{n-1}\\
\one & {\bf c}_1+{\bf d}_2 &({\bf d}_2+\ep_2 {\bf c}_1) {\bf c}_2 & \cdots& ({\bf d}_2+\ep_2 {\bf c}_1) \ep_3\cdots \ep_{n-1} {\bf c}_2 \cdots {\bf c}_{n-1}\\
&\one&{\bf c}_2+{\bf d}_3&\cdots&({\bf d}_3+\ep_3 {\bf c}_2) \ep_4 \cdots \ep_{n-1} {\bf c}_3 \cdots {\bf c}_{n-1} \\
&&\ddots&\ddots &\vdots \\
&&&\one&{\bf c}_{n-1}+{\bf d}_n\\
\end{array}\right]\, . \eaa
}

\begin{prop}{\rm(}\cite{GeKo}{\rm)}\lab{DCeq}. \ The non-Abelian Coxeter-Toda lattice on $G^{u,v}/{\mathbf T}$ with $v=s_{n-1} \cdots s_1$ is equivalent to the following system of equations:
\ba \begin{array}{ll}
\dot{{\bf d}}_i={\bf d}_i {\bf c}_i-{\bf c}_{i-1}{\bf d}_i,\\
\dot{{\bf c}}_i={\bf c}_i {\bf d}_{i+1}-{\bf d}_i {\bf c}_i-\ep_i {\bf c}_{i-1}{\bf c}_i+\ep_{i+1}{\bf c}_i {\bf c}_{i+1}.  & 
\end{array} \label{evolut}
\ea
\end{prop}

Since in the scalar case each system (\ref{evolut}) belongs to a class
 of Coxeter-Toda lattices that were studied in \cite{HKKR,FG1,FG2,GSV}, we will call the system (\ref{evolut}) a {\it non-Abelian Coxeter-Toda lattice}.
\begin{example} If $I=\{n\}$ from Proposition \re{DCeq} we obtain a non-abelian version of the {\it relativistic Toda lattice} (cf. \cite{KSZ} for the scalar case).
\baa \begin{array}{ll}
\dot{{\bf d}}_i={\bf d}_i {\bf c}_i-{\bf c}_{i-1}{\bf d}_i,\\
\dot{{\bf c}}_i={\bf c}_i {\bf d}_{i+1}-{\bf d}_i {\bf c}_i-{\bf c}_{i-1}{\bf c}_i+{\bf c}_i {\bf c}_{i+1}.\\
\end{array}
\eaa
If $I=\{2,\dots,n\}$, the standard non-Abelian Toda lattice (\ref{NAToda}) becomes
\[
\dot{{\bf d}}_i={\bf d}_i {\bf c}_i-{\bf c}_{i-1}{\bf d}_i, \quad
\dot{{\bf c}}_i={\bf c}_i {\bf d}_{i+1}-{\bf d}_i {\bf c}_i\
\]
and after renaming $U_{2i-1}={\bf d}_i, U_{2i}={\bf c}_i$, one obtains the non-Abelian {\it Volterra lattice}
$$
\dot{U}_i=U_i U_{i+1} - U_{i-1} U_i\ .
$$
\end{example}


\section{Conclusion}

Using the noncommutative version of the inverse moment problem, we established that the matrix Weyl function encodes all the information on non-Abelian
Coxeter-Toda lattices. Still, there are some questions that we have not addressed
and that deserve a further investigation. First, in the scalar case, the Hamiltonian structure naturally associated with double Bruhat cells is the Poisson-Lie bracket
on the group $GL_n$ which, when restricted to tri-diagonal matrices, induces the {\em quadratic} Poisson structure for the Toda lattice. However, the Poisson bracket \re{Kost} that leads to the bracket
\re{MMbrack} on rational matrix functions is an analogue of the {\em linear} Poisson structure for the Toda lattice. From this perspective, one expects
a linear Poisson structure \re{Kost} to be replaced by a compatible quadratic one which, in turn, would induce another Poisson bracket on the Weyl function compatible with the one in \re{MMbrack}. More generally, one would like to obtain an analogue of the whole family of compatible Poisson brackets
on the space of rational functions considered in \cite{FG2}. We hope to address this problem in the future.

In addition, it would be interesting  to generalize an approach used in \cite{GSV} to build a noncommutative
cluster algebra structure in the space of rational matrix functions. This should be done in parallel with a transition from a scalar to a noncommutative case in recent works on $Q$-systems by Di Francesco and Kedem \cite{DF,DFNA}.

{\it Acknowledgments.} This work was supported in part by NSF Grant DMS no.~1362801.

\end{document}